\documentclass[letterpaper]{IEEEtran}

\ifCLASSINFOpdf

\else

\fi

\usepackage{cite}
\usepackage{graphicx}
\usepackage{graphics}
\usepackage{epsfig}
\usepackage{epstopdf}
\usepackage{amssymb}
\usepackage[cmex10]{amsmath}
\usepackage{amsthm}
\usepackage{mathrsfs}

\usepackage{ifthen,xkeyval,xfor,amsgen}
\usepackage{float}
\usepackage[utf8]{inputenc}
\usepackage[english]{babel}
\usepackage{cite}

\newtheorem{theorem}{Theorem}
\usepackage{algorithm}
\usepackage{algpseudocode}
\algnewcommand\algorithmicforeach{\textbf{Until :}}
\algnewcommand\algorithmicendif{\textbf{End}}
\algnewcommand\ForEach{\item[ \algorithmicforeach]}
\algnewcommand\EndiFF{\item[ \algorithmicendif]}

\usepackage[margin=0.5in]{geometry}

\begin{document}
%
\title{Robust Beamforming for Cache-Enabled Cloud Radio Access Networks}

\author{Oussama Dhifallah,~\IEEEmembership{Student Member,~IEEE},~Hayssam Dahrouj,~\IEEEmembership{Senior Member,~IEEE},~Tareq Y.~Al-Naffouri,~\IEEEmembership{ Member,~IEEE},~Mohamed-Slim Alouini,~\IEEEmembership{Fellow,~IEEE}\vspace{-8mm}
}


\maketitle

\begin{abstract}
Popular content caching is expected to play a major role in efficiently reducing backhaul congestion and achieving user satisfaction in next generation mobile radio systems. Consider the downlink of a cache-enabled cloud radio access network (CRAN), where each cache-enabled base-station (BS) is equipped with limited-size local cache storage. The central computing unit (cloud) is connected to the BSs via a limited capacity backhaul link and serves a set of single-antenna mobile users (MUs). This paper assumes that only imperfect channel state information (CSI) is available at the cloud. It focuses on the problem of minimizing the total network power and backhaul cost so as to determine the beamforming vector of each user across the network, the quantization noise covariance matrix, and the BS clustering subject to imperfect channel state information, per-BS power constraint, and fixed cache placement assumption. The paper suggests solving such a difficult, non-convex optimization problem using the semi-definite relaxation (SDR) and the S-procedure. The paper uses the $\ell_0$-norm approximation to provide a stationary point using the majorization-minimization (MM) approach. Simulation results show how the cache-enabled network significantly improves the backhaul cost especially at high signal-to-interference-plus-noise ratio (SINR) values as compared to conventional cache-less CRANs.
\end{abstract}
\vspace{-4mm}
\IEEEpeerreviewmaketitle
\section{Introduction}
Cloud radio access network (CRAN) is recognized as a promising network architecture to meet the tremendous increase in data traffic for future networks \cite{CRAN_Andrews, ben_lateiaf, Hayssam_mag}. In CRANs, a central computing unit (cloud) is connected to several BSs through backhaul links which allow joint signal processing of user signals. This allows for effective interference management and significant energy consumption reduction. To meet the increasing demands in data traffic, an increase in the network density is expected, which adds stringent constraints on the backhaul load. This papers addresses the backhaul cost and congestion problem by means of beamforming in cache-enabled networks.

Consider the downlink of a cache-enabled CRAN, where each BS is equipped with a local memory with limited-size. Each cache-enabled BS is connected to the cloud via limited-capacity backhaul link. The central computing unit serves a set of pre-known single-antenna MUs. The paper accounts for imperfect CSI at the cloud and of the quantization noise induced by the employed compression schemes due to the capacity-limited backhaul links. It then optimizes the performance of the system, which is a function of the compression scheme and the beamforming vectors of each user across the network.

The problem considered in this paper is related to the recent state-of-art on CRANs. References \cite{ouss_glob1, weiyu_ds_cs, weiyu_relx, oussama} consider a CRAN scenario and allow all BSs across the network to fetch the requested data from the cloud. In dense data networks, however, such assumption may no longer be feasible because of the high-cost of, and the difficulties in, deploying proliferated high-capacity backhaul networks. As a result, the backhaul congestion is expected to become a limiting factor in futuristic dense networks performance. To truly access the advantages harvested by CRAN, this paper considers the more practical scenario of a cache-enabled CRAN where each BS is equipped with a local memory and connected to the cloud via capacity-limited backhaul links. Such a scenario is particularly related to reference \cite{weiyu_cache}, which addresses a multicast cache-enabled CRAN and formulates the total network cost minimization problem in order to find the beamforming vectors and BS clustering. However, reference \cite{weiyu_cache} assumes perfect CSI available at the cloud and neglects the effect of the employed compression schemes. The problem considered in the current paper is also related to the works in \cite{weiyu_ds_cs,weiyu_relx,oussama}. Reference \cite{weiyu_ds_cs} assumes a conventional (cache-less) CRAN and focuses on solving the total network power minimization problem for both the data-sharing strategy and compression strategy. Reference \cite{weiyu_relx} focuses on solving the utility maximization problem for both the dynamic and static BS clustering. Reference \cite{oussama} assumes hybrid connections between the cloud and the BSs and focuses on minimizing the total network power minimization to jointly design the beamforming vectors, the quantization noise covariance matrix, and the transmit power. The works in \cite{weiyu_ds_cs}, \cite{weiyu_relx} and \cite{oussama}, however, assume a conventional CRAN where all the BSs fetch the requested data from the cloud, and that perfect CSI is available at the cloud, unlike the current paper.

This paper considers a cache-enabled CRAN where the popular content is cached in the local storage of the BSs. When the content requested by a MU is available in the local cache, it is directly transmitted from the cache to the user with no need for backhauling and compression. When the data is not available in the local cache, the BS fetches the data from the cloud which performs joint precoding and compression. Such pre-processing of the data introduces quantization noise, which degrades the system performance. The paper formulates the problem of minimizing the total network power and backhaul cost subject to per-BS power constraint, quality of service constraints, per-BS backhaul capacity constraint, and imperfect CSI constraints. The paper highlight is that it solves such a non-convex  optimization problem using the semi-definite relaxation (SDR) \cite{Zhengrobustlinear}, and the S-Procedure method \cite{convex_opt}. The proposed algorithm then utilizes one particular approximation of the $l_0$ norm to produce a stationary point using the majorization-minimization (MM) approach \cite{mm_algo}. Simulation results particularly show how the cache-enabled network significantly improves the backhaul cost as compared to conventional cache-less CRANs especially at high SINR values.
\vspace{-2mm}
\section{System Model and Problem }
\label{sec:un}
\subsection{System Model}
Consider the downlink of a cache-enabled CRAN, where the cloud is connected to $B$ BSs through capacity-limited backhaul links. Each BS serves $U$ single-antenna MUs, and is equipped with a local cache with limited storage. For the simplicity of analysis, we assume that each BS is equipped with single-antenna, all the files have the same size, and the knowledge of cache and request status are known at the cloud. 


Let $\mathcal{B}=\lbrace 1,..,B\rbrace$ be the set of BSs connected to the cloud, and $\mathcal{U}=\lbrace 1,..,U\rbrace$ be the set of users across the network. Let the beamforming vector from the set of all BSs $\mathcal{B}$ to user $u$ ($u \in \mathcal{U}$) be ${\bf w}_{u}=[w_{1u},~w_{2u},..,~w_{Bu}]^T  \in \mathbb{C}^{B}$, and let the channel vector from $\mathcal{B}$ to $u$ be ${\bf h}_{u}=[h_{1u},~h_{2u},..,~h_{Bu}]^T  \in \mathbb{C}^{B}$. Further, let the quantization noise vector be ${\bf v}$ $=[v_{1},~v_{2},\cdots,~v_{B}] \in \mathbb{C}^{B}$, where ${\bf v}$ is non-uniform white Gaussian with diagonal covariance matrix ${\bf Q}$ $\in \mathbb{C}^{B\times B}$ of diagonal entries $q_{b}$ ($q_{b} \geq 0,~\forall b\in\mathcal{B}$). Each cache-enabled BS $b$ has a local storage with size $F_b$, which allows to describe the cache placement through the cache placement matrix ${\bf P}\in\mathbb{R}^{B\times F}$, where $p_{bf}=1$ when the content $f$ is cached at BS $b$, and $p_{bf}=0$ otherwise, i.e. $\sum_{f=1}^{F}p_{bf} < F_b,~\forall b\in\mathcal{B}$. The matrix ${\bf P}$ is kept fixed throughout the current paper. \footnote{Optimizing the cache placement is left for future work, as it increases the complexity of our complex problem.} The received signal $y_u \in \mathbb{C}$ at user $u$ can be written as follows:
\begin{align}
y_u&={\bf h}_{u}^{H}{\bf w}_{u}s_u+ \sum\limits_{u^\prime\in \mathcal{U}_{u}}^{}~{\bf h}_{u}^{H}{\bf w}_{u^\prime }s_{u^\prime}+\sum\limits_{b\in \mathcal{B}}^{}~ \alpha_b h_{bu}^{*}v_{b}+n_u,
\label{re_wav}
\end{align}
where $\mathcal{U}_{u}=\mathcal{U}\setminus \lbrace u\rbrace$, $s_u$ denotes the transmitted data symbol for user $u$, $n_u~\sim~\mathcal{CN}(0,\sigma_u^2)$ denotes the additive white Gaussian noise which is independent from the transmitted data symbols $s_u$ and the quantization noise $v_{b}$, and where $\alpha_b=0$ when the requested content  by all the users across the network is cached at BS $b$, and $1$ otherwise, since if the content $f_u$ requested by user $u$ from BS $b$ is available in the local cache of BS $b$, the BS transmits $f_u$ directly without backhauling and compression. When $f_u$ is not in the cache of BS $b$, it needs to be retrieved from the cloud. The focus of this paper is on the scenario when all the BSs across the network cache the same popular data and each user request the same data from the active BSs.

\subsection{Backhaul Cost and Power Model}
This paper considers the problem of minimizing the total network cost, which consists of the backhaul cost of fetching the contents from the cloud, in addition to the network power consumption. Firstly, the backhaul cost of fetching the contents from the cloud is assumed to be proportional to the transmission rate of user $u$ since, without loss of generality, the data rates of fetching the requested data from the cloud need to be as large as the content delivery. The backhaul cost associated with an active BS $b$ serving user $u$ can therefore be expressed as follows:
\begin{align}\label{back_cost}
B_{bu} = \begin{cases}
\textnormal{log}_2(1+\delta_u) &\mbox{if } p_{bf_u}= 0~\mbox{and}~|w_{bu}|^2 > 0 \\
0 & \mbox{Otherwise }
\end{cases},
\end{align}
where $\delta_u$ denotes the target SINR to be achieved by user $u$, and $f_u$ is the file content requested by user $u$.

Since each active BS is assumed to serve all the users across the network, the power consumption of BS $b$ can be written as:
\begin{align}
\hspace{-3mm} P_b = \begin{cases}
\frac{P_{b,t}}{\nu_b} +P_{b,a} &\mbox{if } P_{b,t} > 0 \\
P_{b,s} & \mbox{if } P_{b,t}=0
\end{cases},
\end{align}
where $P_{b,t}=\left(\sum_{u\in\mathcal{U}}^{}|w_{bu}|^2+\alpha_b q_{b}\right)$ denotes the transmit power, $\nu_b$ denotes the power amplifier efficiency, $P_{b,a}$ denotes the power consumed by BS $b$ in the active mode, and $P_{b,s}$ denotes the power consumed by BS $b$ in the sleep mode. The total power consumption across the network can then be written as:
\begin{align}\label{tot_pow}
\hspace{-2mm} P_t=\sum\limits_{b\in\mathcal{B}}^{} \left\lbrace \frac{\sum\limits_{u\in\mathcal{U}}^{}|w_{bu}|^2+\alpha_b q_{b}}{\nu_b}+ \Big|\Big| P_{b,t} \Big|\Big|_0 P_{rb}+P_{b,s}\right\rbrace,
\end{align}
where $P_{rb}=P_{b,a}-P_{b,s}$ denotes the relative power consumption.

\subsection{Problem Formulation}
The problem considered in this paper consists of minimizing the total network cost subject to per-BS power constraints, per-BS backhaul capacity constraints, quality of service constraints, and CSI error constraints. The paper assumes that the channel errors are bounded by an elliptical region. The true channel vector ${\bf h}_{u}$ of user $u$ can therefore be written as follows:
\begin{equation}
\begin{aligned}\label{csi_err}
{\bf h}_{u}=\tilde{\bf h}_{u}+{\bf e}_{u},~\forall~u\in \mathcal{U},
\end{aligned}
\end{equation}
where ${\bf e}_{u}$ denotes the CSI error vector of user $u$ that satisfies the following elliptical constraints:
\begin{equation}
\begin{aligned}\label{ell_set}
{\bf e}_{u}^H{\bf E}_{u}{\bf e}_{u} < 1,
\end{aligned}
\end{equation}
where $\tilde{\bf h}_{u}$ denotes the estimated channel vector, and ${\bf E}_{u}$ is a known positive definite matrix that best measures the accuracy of the CSI.  Assume that the user symbols have unit power, i.e., $E(|s_u|^2 )= 1$,$~\forall~u\in \mathcal{U}$, and are independent from each other, from the quantization noise, and from the additive noise. The SINR of user $u$ can then be written as follows:
\begin{equation}
\begin{aligned}
\Gamma_u=\frac{\left|{\bf h}_{u}^{H}{\bf w}_{u} \right|^2}{\sum\limits_{u^\prime\in \mathcal{U}_{u}}^{}\left|{\bf h}_{u}^{H}{\bf w}_{{u^\prime}}\right|^2+\sum\limits_{b\in \mathcal{B}}^{}~ \alpha_b |h_{bu}|^2 q_{b}+\sigma_u^2}.
	\label{sinr}
\end{aligned}
\end{equation}
The optimization problem considered in this paper is also subject to per-BS power constraint which can be written as follows:
\begin{equation}
\begin{aligned}
\sum\limits_{u\in\mathcal{U}}^{}~|w_{bu}|^2+\alpha_b q_{b} \leq P_{b},~\forall~b\in \mathcal{B}.
	\label{power_const}
\end{aligned}
\end{equation}
If the requested data from a MU is not available in the local cache of the BS, the data is fetched from the cloud through the limited-capacity backhaul link by means of compression, which is assumed to be independent among users for the sake of simplicity. If the requested data is available at the BS local cache, no compression is needed. Therefore, the beamforming vector associated with user $u$, the quantization noise level $q_{b}$, and the backhaul capacity of BS $b$, $C_b$, are related as follows:
\begin{equation}
\label{rate_distortion1}
\textnormal{log}_2\Bigg(1+\frac{\sum\limits_{u\in\mathcal{U}}^{}(1-p_{bf_u})|w_{bu}|^2}{q_{b}}\Bigg) \leq C_b.
\end{equation}
This paper minimizes the total network cost, denoted by $C_N$, which consists of the total power consumption (\ref{tot_pow}), and the backhaul cost (\ref{back_cost}):
\begin{align}
&C_N=\sum\limits_{b\in\mathcal{B}}^{}\Bigg\lbrace \frac{\sum\limits_{u\in\mathcal{U}}^{}|w_{bu}|^2+\alpha_b q_{b}}{\nu_b}+\left|\left|\sum\limits_{u\in\mathcal{U}}^{}|w_{bu}|^2+\alpha_b q_{b}\right|\right|_0 P_{rb}\nonumber\\	
&+\sum\limits_{u\in\mathcal{U}}^{}\Big|\Big||w_{bu}|^2\Big|\Big|_0 \left( 1-p_{bf_u}\right)R_u \Bigg\rbrace
\end{align}
where $R_u=\textnormal{log}_2(1+\delta_u)$ denotes the target rate of user $u$.
The overall optimization problem can therefore be formulated as:
\begin{align}\label{glob_form}
\underset{{\bf w},{\bf q}}{\operatorname{min}}&~~~~~~C_N\\
\text{subject to}&~~~~~~\mathrm{Constraints}~(\ref{csi_err}),(\ref{ell_set}),~(\ref{power_const}),~\mathrm{and}~(\ref{rate_distortion1})\nonumber\\
&\tilde{\Gamma}_u=\frac{\left|{\bf h}_{u}^{H}{\bf w}_{u} \right|^2}{\sum\limits_{u^\prime\in \mathcal{U}_{u}}^{}\left|{\bf h}_{u}^{H}{\bf w}_{{u^\prime}}\right|^2+{\bf h}_{u}^{H}{\bf Q}{\bf h}_{u}+\sigma_u^2} \geq \delta_u,~\forall~u, \nonumber
\end{align}
where the optimization is over the beamforming vectors ${\bf w}$ and the quantization noise vector ${\bf q}$, and where the SINR expressions $\tilde{\Gamma}_u$ in (\ref{glob_form}) are equivalent to the expressions in (\ref{sinr}). Problem (\ref{glob_form}) is difficult to solve due to the non-convexity of the cost function and the infinite set of possible CSI errors. This paper proposes solving (\ref{glob_form}) using the SDR and the S-Procedure methods \cite{Zhengrobustlinear,convex_opt}. It uses one $\ell_0$-norm approximation to provide a stationary point using the MM approach. Simulations results suggest that the proposed algorithm provides a significant performance improvement as compared to cache-less networks.
\section{Algorithm}
\label{sec:deux}
\subsection{Semi-Definite Programming (SDP) Reformulation}
Define the rank-one matrix ${\bf W}_{u}$ as ${\bf W}_{u}={\bf w}_{u}{\bf w}^{H}_{u}$, $\forall~u\in\mathcal{U}$. The S-Procedure \cite{convex_opt} and the rank-one SDR approach \cite{Zhengrobustlinear} are then used to derive the steps of our proposed algorithm. After dropping the non-convex rank-one constraints, the minimization problem (\ref{glob_form}) can be reformulated as follows:
\begin{align}\label{sdp_form}
&\underset{{\bf W}_{u},{\bf Q},\lambda_{u}}{\operatorname{min}}~~~~~~~~\hat{C}_N\\
&\text{subject to}\sum_{u\in\mathcal{U}}^{}~\mathrm{Tr}\left({\bf A}_{b}{\bf W}_{u}\right)+\alpha_b \mathrm{Tr}\left({\bf A}_{b}{\bf Q}\right) \leq P_{b}\nonumber\\
&~~~~\sum_{u\in\mathcal{U}}^{}(1-p_{bf_u})\mathrm{Tr}\left({\bf A}_{b}{\bf W}_{u}\right)-(2^{C_b}-1) \mathrm{Tr}\left({\bf A}_{b}{\bf Q}\right) \leq 0\nonumber\\
&~~~~~{\boldsymbol \Delta}_{u} \succeq 0,~\lambda_{u} \geq 0,{\bf W}_{u} \succeq 0,{\bf Q} \succeq 0,{\bf Q}~\mathrm{is~diagonal},~\forall~u, \nonumber
\end{align}
\begin{figure*}
\begin{align} \label{cost_sdp}\small
\hat{C}_N&=\sum\limits_{b\in\mathcal{B}}^{}\Bigg\lbrace \frac{\sum\limits_{u\in\mathcal{U}}^{}\mathrm{Tr}\left({\bf A}_{b}{\bf W}_{u}\right)+\mathrm{Tr}\left({\bf A}_{b}{\bf Q}\right)}{\nu_b}+\Bigg|\Bigg|\sum\limits_{u\in\mathcal{U}}^{}\mathrm{Tr}\left({\bf A}_{b}{\bf W}_{u}\right)+\mathrm{Tr}\left({\bf A}_{b}{\bf Q}\right)\Bigg|\Bigg|_0 P_{rb}+\sum\limits_{u\in\mathcal{U}}^{}\Big|\Big|\mathrm{Tr}\left({\bf A}_{b}{\bf W}_{u}\right)\Big|\Big|_0 R_u \left( 1-p_{bf_u}\right)\Bigg\rbrace.
\end{align}
\hrulefill
\end{figure*}
\hspace{-2.4mm} where the optimization is over the beamforming matrices ${\bf W}_{u}$, the quantization noise covariance matrix ${\bf Q}$, and the introduced variables $\lambda_{u}$, and where $\hat{C}_N$ is given in (\ref{cost_sdp}), $\hat{R}^{m}_b=\gamma^{m}_b \tilde{R}^{m}_b$, $\hat{a}^{m}_b=\frac{a^{m}_b}{\gamma^{m}_b}$, ${\bf A}_{b}$ denotes the diagonal matrix with $1$ at the main diagonal entry $b$ and zeros otherwise, and the matrix ${\boldsymbol \Delta}_{u}$ is defined as:
\begin{align}
{\boldsymbol \Delta}_{u}=
\begin{bmatrix}
    {\bf G}_{u}+\lambda_{u}{\bf E}_{u} & {\bf G}_{u}\tilde{\bf h}_{u}  \\
    \tilde{\bf h}^{H}_{u}{\bf G}_{u} & \tilde{\bf h}^{H}_{u}{\bf G}_{u}\tilde{\bf h}_{u}-\sigma_u^2-\lambda_{u}
\end{bmatrix},
\end{align}
and where
\begin{align}
{\bf G}_{u}=\frac{1}{\delta_u}{\bf W}_{u}-\sum\limits_{u^\prime\in \mathcal{U}_{u}}^{}{\bf W}_{u^\prime}-{\bf Q}.
\end{align}
While the feasibility set of the optimization problem (\ref{sdp_form}) is convex, the cost function (\ref{cost_sdp}) is still not convex. The paper, therefore, proposes determining an approximate solution to the relaxed optimization problem (\ref{sdp_form}) by first approximating the $\ell_0$-norm and then by using the MM algorithm.

\subsection{Majorization-Minimization Approach}
This section first suitably approximates the cost function (\ref{cost_sdp}) so as to pave the way for the MM algorithm steps. Consider the following $\ell_0$-norm approximation:
\begin{align}\label{lim_app}
||x||_0=I(x>0)=\lim_{\epsilon\to 0} \frac{\textnormal{log}\left(1+\epsilon^{-1}x\right)}{\textnormal{log}(1+\epsilon^{-1})},~x \geq 0.
\end{align}
The cost function (\ref{cost_sdp}) can then be approximated as follows:
\begin{align}\label{obj_app}
&\tilde{C}_N=\sum\limits_{b\in\mathcal{B}}^{}\Bigg\lbrace \frac{1}{\nu_b}\left(\sum\limits_{u\in\mathcal{U}}^{}\mathrm{Tr}\left({\bf A}_{b}{\bf W}_{u}\right)+\mathrm{Tr}\left({\bf A}_{b}{\bf Q}\right) \right)\nonumber\\
&+\lambda_{\epsilon}\textnormal{log}\left(1+\epsilon^{-1}\left\lbrace \sum\limits_{u\in\mathcal{U}}^{}\mathrm{Tr}\left({\bf A}_{b}{\bf W}_{u}\right)+\mathrm{Tr}\left({\bf A}_{b}{\bf Q}\right) \right\rbrace \right) P_{rb}\nonumber\\
&+\lambda_{\epsilon}\sum\limits_{u\in\mathcal{U}}^{}\textnormal{log}\left(1+\epsilon^{-1} \mathrm{Tr}\left({\bf A}_{b}{\bf W}_{u}\right) \right) R_u \left( 1-p_{bf_u}\right)\Bigg\rbrace.
\end{align}
The MM algorithm can now be used to find a stationary point to the obtained optimization problem, i.e., the problem obtained by replacing $\hat{C}_N$ with $\tilde{C}_N$ in (\ref{sdp_form}), by first finding a surrogate function that majorizes $\tilde{C}_N$. Then, it iteratively minimizes the obtained function until a local optimal solution of the optimization problem with cost function (\ref{obj_app}) is reached.
\begin{figure*}
\begin{align} \label{const_prev}\small
&\eta^{m}_b=\frac{1}{\nu_b}+\frac{\lambda_{\epsilon}P_{rb}}{\epsilon+\sum\limits_{u\in\mathcal{U}_{b}}^{}\mathrm{Tr}\left({\bf A}_{b}{\bf W}^{m}_{u}\right)+\mathrm{Tr}\left({\bf A}_{b}{\bf Q}^{m}\right)},~~\beta_{bu}^{m}=\frac{R_u(1-p_{bf_u})\lambda_{\epsilon}}{\epsilon+\mathrm{Tr}\left({\bf A}_{b}{\bf W}^{m}_{u}\right)}.
\end{align}
\hrulefill
\end{figure*}
\begin{theorem}\label{them_32}
The surrogate function that majorizes the function (\ref{obj_app}) at iteration $m+1$ is given by:
\begin{align}\small
&C_N^{m}({\bf W}_{u},{\bf Q})=\sum\limits_{b\in\mathcal{B}}^{}\eta^{m}_b \left(\sum\limits_{u\in\mathcal{U}}^{}\mathrm{Tr}\left({\bf A}_{b}{\bf W}_{u}\right)+\mathrm{Tr}\left({\bf A}_{b}{\bf Q}\right) \right)\nonumber\\
&+\sum\limits_{b\in\mathcal{B}}^{}\sum\limits_{u\in\mathcal{U}}^{}\beta_{bu}^{m}\mathrm{Tr}\left({\bf A}_{b}{\bf W}_{u}\right)+c({\bf W}^{m}_{u},{\bf Q}^{m}).
\end{align}
where $\eta^{m}_b$, $\beta_{bu}^{m}$ are given in (\ref{const_prev}) and $c({\bf W}^{m}_{u},{\bf Q}^{m})$ is a constant which only depends on the matrices ${\bf W}^{m}_{u}$ and the quantization noise covariance matrix ${\bf Q}^{m}$ of the previous iteration.
\end{theorem}
\begin{proof}
Let ${\bf W}^{m}_{u}$ and ${\bf Q}^{m}$ denote the beamforming matrix associated with user $u$ and the quantization noise covariance matrix of the previous iteration of the MM algorithm, respectively. Based on the fact that the function $x\rightarrow\textnormal{log}(1+\epsilon^{-1} x)$ is a concave function on the interval $[0~+\infty[$, for $\epsilon > 0$, we have
\begin{align}\label{con_log2_th21}
&\hspace{-2mm}\textnormal{log}\left(1+\epsilon^{-1}\left\lbrace \sum\limits_{u\in\mathcal{U}}^{}\mathrm{Tr}\left({\bf A}_{b}{\bf W}_{u}\right)+\mathrm{Tr}\left({\bf A}_{b}{\bf Q}\right) \right\rbrace \right)\nonumber\\
&\hspace{-2mm}\leq \frac{ \sum\limits_{u\in\mathcal{U}}^{}\mathrm{Tr}\left({\bf A}_{b}{\bf W}_{u}\right)+\mathrm{Tr}\left({\bf A}_{b}{\bf Q}\right)}{\epsilon+ \sum\limits_{u\in\mathcal{U}}^{}\mathrm{Tr}\left({\bf A}_{b}{\bf W}^{m}_{u}\right)+\mathrm{Tr}\left({\bf A}_{b}{\bf Q}^{m}\right)}+c_{b1}({\bf W}^{m}_{u},{\bf Q}^{m}),
\end{align}
and
\begin{align}\label{con_log2_th22}
\hspace{-3mm} \textnormal{log}\left(1+ \frac{\mathrm{Tr}\left({\bf A}_{b}{\bf W}_{u}\right)}{\epsilon} \right) \leq \frac{ \mathrm{Tr}\left({\bf A}_{b}{\bf W}_{u}\right) }{\epsilon+ \mathrm{Tr}\left({\bf A}_{b}{\bf W}^{m}_{u}\right) }+c_{bu2}({\bf W}^{m}_{u}),
\end{align}
where $c_{b1}({\bf W}^{m}_{u},{\bf Q}^{m})$ and $c_{bu2}({\bf W}^{m}_{u})$ are constants that only depend on the beamforming matrices ${\bf W}^{m}_{u}$ and the quantization noise covariance matrix ${\bf Q}^{m}$ of the previous iteration. Based on (\ref{con_log2_th21}) and (\ref{con_log2_th22}), we have
\begin{equation}
\begin{aligned}\label{con_log3_th23}
\tilde{C}_N \leq C_N^{m}({\bf W}_{u},{\bf Q}),
\end{aligned}
\end{equation}
where
\begin{align}\label{con_log3_th24}
c({\bf W}^{m}_{u},{\bf Q}^{m})&=\lambda_{\epsilon}\sum\limits_{b\in\mathcal{B}}^{}c_{b1}({\bf W}^{m}_{u},{\bf Q}^{m}) P_{rb}\nonumber\\
&+\lambda_{\epsilon}\sum\limits_{b\in\mathcal{B}}^{}\sum\limits_{u\in\mathcal{U}}^{}c_{bu2}({\bf W}^{m}_{u}) R_u \left( 1-p_{bf_u}\right).
\end{align}
This completes the proof of theorem \ref{them_32}.
\end{proof}

Using the above theorem, the MM approach solves the following optimization problem at the iteration $m+1$:
\begin{align}\label{MM_fin_opt}
&\underset{{\bf W}_{u},{\bf Q},\lambda_{u}}{\operatorname{min}}~~~~~~~~C_N^{m}({\bf W}_{u},{\bf Q})\\
&s.t.\sum\limits_{u\in\mathcal{U}}^{}~\mathrm{Tr}\left({\bf A}_{b}{\bf W}_{u}\right)+\alpha_b \mathrm{Tr}\left({\bf A}_{b}{\bf Q}\right) \leq P_{b},~\forall~b\in\mathcal{B}\nonumber\\
&~~~~\sum\limits_{u\in\mathcal{U}}^{}(1-p_{bf_u})\mathrm{Tr}\left({\bf A}_{b}{\bf W}_{u}\right)-(2^{C_b}-1) \mathrm{Tr}\left({\bf A}_{b}{\bf Q}\right) \leq 0\nonumber\\
&~~~~~{\boldsymbol \Delta}_{u} \succeq 0,~\lambda_{u} \geq 0,{\bf W}_{u} \succeq 0,{\bf Q} \succeq 0,{\bf Q}~\mathrm{is~diagonal},~\forall~u\in\mathcal{U} \nonumber
\end{align}
The above optimization problem is an SDP. Therefore, it can be solved using efficient numerical algorithms \cite{convex_opt}.

\vspace{-2mm}
\subsection{Proposed Iterative Algorithm}
The overall algorithm used to solve the original optimization problem now iterates between two levels. At the first level, it solves the optimization problem (\ref{MM_fin_opt}). Then, it updates $\eta^{m+1}_b$ and $\beta^{m+1}_b$. Such algorithm is guaranteed to converge as shown in the following theorem.
\begin{theorem}\label{them_33}
The proposed iterative algorithm is guaranteed to converge to a stationary point of the original optimization problem (\ref{glob_form}) as $\epsilon$ tends to 0.
\end{theorem}
The proof of the above theorem is based on the fact that the proposed iterative algorithm is equivalent to an MM algorithm which is guaranteed to converge to the stationary point solutions \cite{conv_mm}. While the proposed rank-1 relaxation might not always lead to a rank-one solution, a rank-one solution can be achieved by well-known randomization techniques \cite{convex_opt}. 
\subsection{Computational Complexity Analysis}
The implementation of the proposed iterative algorithm requires to solve the SDP problem (\ref{MM_fin_opt}) with $c=(2B+3U+2)$ SDP constraints and $v=(2U+1)$ SDP variables at each iteration. The second step of the proposed algorithm consists of updating $\eta^{m+1}_b$ and $\beta^{m+1}_b$. The computational complexity comes mainly, therefore, from solving the SDP problem. If the obtained solution does not satisfy the relaxed rank-one constraints, Gaussian randomization techniques \cite{convex_opt} can be applied to estimate a rank-one solution, which adds to the overall algorithmic complexity.


\vspace{-2mm}
\section{Simulation results}
\label{sec:trois}
\begin{figure}
\begin{center}
\rotatebox{0}{\scalebox{0.24}{\includegraphics{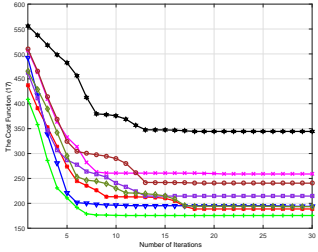}}}
\centering
\caption{Convergence behavior of the Iterative Relaxed MM Algorithm.}
\end{center}
\label{conv_behav1}
\end{figure}
\begin{figure}
\begin{center}
\rotatebox{0}{\scalebox{0.33}{\includegraphics{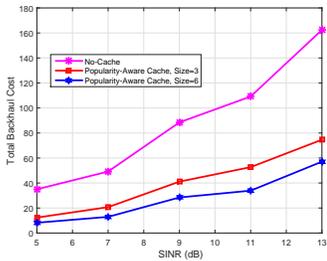}}}
\centering
\caption{The total backhaul cost as function of the target SINR.}
\end{center}
\label{sinr_var_3}
\end{figure}

This section provides simulation examples to illustrate the performance of the proposed algorithm. It considers a cache-enabled CRAN scenario formed by $B=14$ single-antenna cache-enabled BSs, where each BS is connected to $U=6$ single-antenna MUs. The BSs and MUs are uniformily and independently distributed in the square region $[0~1000] \times [0~1000]$ meters. Further, the estimated channel vectors are generated using Rayleigh fading component and a distance-dependent path loss, modeled as $L(d_{bu})= 128.1+37.6\textrm{log}_{10}(d_{bu})$, where $d_{bu}$ denotes the distance between BS $b$ and user $u$ in kilometers. Each user across the network randomly requests one content from the BSs according to the content popularity, modeled as Zipf distribution with skewness parameter 1. The noise power spectral density is set to $\sigma_u^2=-98~\textrm{dBm}$ $\forall u$. We set the maximum transmit power of BS $b$ to $P_{b}=1$ Watt, the per-BS backhaul capacity limit to $C_b=10$ Mbps, $\forall b$, the relative power consumption to $P_{rb}=38$ Watts, $\nu_b=2.5$ and the accuracy matrix ${\bf E}_{u}=\frac{1}{a}{\bf I}_{B}$ where $a > 0$. $\eta^{0}_b$ and $\beta^{0}_b$~$\forall b\in\mathcal{B}$ are initially all set to 1. All provided simulations are rank-one solutions.
%


First, the SINR target is set to $\delta_u=10\mathrm{dB}$~$\forall u\in\mathcal{U}$, the positive constant $a$ to $0.01$ and $\epsilon=10^{-6}$. Figure \ref{conv_behav1} shows the convergence behavior of the proposed iterative algorithm for different channel realizations. It can be noticed that the proposed algorithm converges for all the considered channel realizations. Figure \ref{conv_behav1} further shows that the proposed algorithm has a reasonable convergence speed (around 20 iterations) for the considered channel realizations. It is especially remarkable that the cost function, i.e., function (\ref{obj_app}), is always driven downhill.


 Consider now that users across the network request different contents, where the network is formed by $B=12$ BSs and $U=8$ MUs, where half of the scheduled users request a common content, and the other half randomly request one content. Figure \ref{sinr_var_3} shows the total backhaul cost versus the target SINR. It can be noticed that the cache-enabled network significantly reduces the backhaul cost especially at high SINR, as compared to the network without cache. It is remarkable how increasing the cache size significantly reduces the total backhaul cost, which enlightens the role of proactive caching in future networks.
%

\vspace{-3mm}
\section{Conclusion}
\label{sec:quatre}
Proactively caching popular content is expected to play a major role in alleviating backhaul congestion problems. This paper considers the downlink of a cache-enabled CRAN, where each cache-enabled single-antenna BS serves a pre-known set of single-antenna MU. The paper assumes that only imperfect CSI is available at the cloud, and each BS is connected to the cloud through limited-capacity backhaul link. The paper provides an iterative algorithm to solve the total network power and backhaul cost minimization problem. Simulation results show that the cache-enabled network significantly reduces the backhaul cost especially at high SINR as compared to networks without cache.
\vspace{-3mm}
\bibliographystyle{IEEEtran}
\bibliography{IEEEabrv,reference}

\ifCLASSOPTIONcaptionsoff
  \newpage
\fi

\end{document}